\documentclass[a4paper,11pt]{article}
\usepackage[utf8]{inputenc}

\usepackage{geometry}
\usepackage{xcolor}
\usepackage{amsmath, array, amssymb, amsfonts,amsthm}

\usepackage[french, english]{babel}
\usepackage[all]{xy}

\usepackage{lmodern}

\usepackage{sectsty}
\sectionfont{\centering}
\subsectionfont{\centering}
\subsubsectionfont{\centering}

\usepackage{booktabs}
\usepackage{caption}
\usepackage{dsfont}
\usepackage{mathtools}

\usepackage[colorlinks=true, implicit=false]{hyperref}

\usepackage{textcomp}

\usepackage{multicol}
\usepackage[square,numbers,sort,compress,semicolon,merge]{natbib}
\usepackage{notoccite}
\newcommand{\defeq}{\vcentcolon=}
\newcommand{\rdefeq}{=\vcentcolon}
\renewcommand\P{\mathcal{P}}
\newcommand\M{\mathcal{M}}
\newcommand\R{\mathcal{R}}
\newcommand\doubleC{\mathbb{C}}
\renewcommand\H{\mathcal{H}}
\newcommand\A{\mathcal{A}}

\newcommand\U{\mathcal{U}}
\newcommand{\circoverset}[1]{\mathring{#1}}
\newcommand{\domega}{{\circoverset{\omega}}}
\newcommand{\dOmega}{\circoverset{\Omega}}
\newcommand{\bgamma}{{\bar\gamma}}
\newcommand{\bchi}{{\bar\chi}}
\newcommand{\bu}{{\bar u}}
\newcommand{\bc}{{\bar c}}

\newcommand{\homega}{{\widehat{\omega}}}
\newcommand{\hOmega}{{\widehat{\Omega}}}
\newcommand{\halpha}{{\widehat{\alpha}}}
\newcommand{\hpsi}{{\widehat{\psi}}}
\newcommand{\hPsi}{{\widehat{\Psi}}}
\newcommand{\hA}{{\widehat{A}}}
\newcommand{\hF}{{\widehat{F}}}
\newcommand{\hD}{{\widehat{D}}}
\newcommand{\tphi}{\widetilde{\phi}}
\newcommand\rarrow{\rightarrow}

\newcommand\LieH{\mathfrak{h}}

\newcommand\w{\wedge}
\newcommand\Em{^{-1}}

\DeclareMathOperator{\Iso}{Iso}

\DeclareMathOperator{\Aut}{Aut}
\DeclareMathOperator{\Tr}{Tr}

\DeclareMathOperator{\Ad}{Ad}
\DeclareMathOperator{\id}{id}

\newtheorem{thm}{Theorem}

\newtheorem{prop}[thm]{Proposition}
\theoremstyle{definition}
\newtheorem{definition}[thm]{Definition}

\hypersetup{
pdfauthor={J. François, S. Lazzarini, T. Masson},
pdftitle={Nucleon spin decomposition and differential geometry},
pdfsubject={},
pdfcreator={pdflatex},
pdfproducer={pdflatex},
pdfkeywords={}
}

\begin{document}

\title{Nucleon spin decomposition and differential geometry}
\author{J. François, S. Lazzarini, T. Masson}
\date{}
\maketitle
\begin{center}
Centre de Physique Théorique \\
Aix Marseille Université \& Université de Toulon \&  CNRS  (UMR 7332)\\
 Case 907, 13288 Marseille, France
\end{center}

\begin{abstract}
In the last few years, the so-called \emph{Chen et al. approach} of  the nucleon spin decomposition has been widely discussed and elaborated on. In this letter we propose a genuine differential geometric understanding of this approach. We mainly highligth its relation to the ``dressing field method'' we  advocated in \cite{GaugeInv}. We are led to the conclusion that the claimed gauge-invariance of the \emph{Chen et al}. decomposition is actually unreal. 
\end{abstract}


\vspace{1mm}


\vspace{2mm}

\begin{multicols}{2}

\section*{Introduction}  

Since more than 25 years, state of the art experiments  have shown that the quarks spin contribute to about a third of the nucleon spin, while the gluons spin have been deemed to contribute little. Actually  recent results \cite{deFlorian-et-al14} tend to put into question this claim, the gluons may significantly contribute after all. The remaining fraction should then be attributed to the orbital angular momentum (OAM) of the quarks and gluons. See \cite{LorceGeomApproach, Wakamatsu14} and references therein. Anyway these experimental breakthrough have revived theoretical inquiries on the nucleon spin decomposition in QCD.

There are several decompositions each with advantages and drawbacks (see \cite{Leader-Lorce} for details), the most salient concerning a trade between partonic interpretation and gauge-invariance.  The Belifante decomposition splits the angular momentum of the nucleon $J_\text{n}$ as two contributions from the angular momentum of the quarks and gluons: $J_\text{n}=J_\text{Bel}^\text{q}+J_\text{Bel}^\text{g}$. Both terms are gauge-invariant thus observable/measurable, but give no further understanding of the contributions of the quarks and gluons spins and OAM. The Ji decomposition provides a further split (up to an exact term) 
of the quarks contribution into spin and OAM: $J_\text{n}=S_\text{Ji}^\text{q}+L_\text{Ji}^\text{q}+J_\text{Ji}^\text{g}$. Each term is gauge-invariant but  
 gluons spin and OAM are still not disentangled. The Jaffe-Monohar decomposition improves on the latter by splitting the gluons contribution: 
 $J_\text{n}=S_\text{JM}^\text{q}+L_\text{JM}^\text{q}+S_\text{JM}^\text{g}+L_\text{JM}^\text{g}$. Nevertheless here only $S_\text{JM}^\text{q}$ is gauge-invariant. Up to this point it seemed there was a conflict between a clear partonic interpretation of the contributions and their gauge-invariance.

A few years ago Chen \& al. \cite{Chen-et-al08, Chen-et-al09} proposed a new decomposition based on the ansatz that the gauge potential can be split into pure gauge and physical  parts: $J_\text{n}=S_\text{Chen}^\text{q}+L_\text{Chen}^\text{q}+S_\text{Chen}^\text{g}+L_\text{Chen}^\text{g}$. Each piece appears now to be gauge-invariant so that the decomposition seems to finally give satisfaction.  
However the non-manifest Lorentz covariance of the approach  raised questions. 
Several authors addressed this issue. Wakamatsu \cite{Wakamatsu11}  developed a manifestly Lorentz covariant version of the Chen \& al. approach. Then Lorcé \cite{LorceGeomApproach}  attempted a geometrical interpretation  and further critically discussed the issue of Lorentz covariance.  Moreover he identified, following Stoilov \cite{Stoilov}, a problem of non-uniqueness (in the splitting of the gauge potential)  referred to as a ``Stueckelberg symmetry''. See also \cite{LorceWilson, LorceCanForm, Lorce14, Leader-Lorce}.

Here we propose a genuine differential geometric basis to the Chen \& al. ansatz and highlight its relation to the ``dressing field method'' to construct gauge-invariants presented in \cite{GaugeInv}. In doing so we deem that we circumvent the discussion about Lorentz covariance, clarify the status of the so-called Stueckelberg symmetry and prevent possible misconceptions. 
The paper is organized as follows. In section 1 we recall the basics of the bundle geometry underlying  gauge theories. In section 2 we give the simplest version of the dressing field method and stress important interpretive points. In section 3 we propose a differential geometric formulation of the Chen \& al. ansatz and apply the dressing field method. Both the global and local constructions are given so as to keep in touch with the pre-existing literature. The crucial question of the arbitrariness on the choice of the dressing field is addressed. In section 4 we gather our comments and then conclude.

\section{The geometry of gauge theories}  
\label{The geometry of gauge theories}  

The geometry of fibered spaces is the natural framework of gauge theories (see e.g \cite{DeAzc-Izq, Bertlmann, GockSchuck, Nakahara} for pedagogical introductions). One starts with a principal bundle $\P(\M, H)$ over space-time $\M$ with structure group $H$. A choice of connection $1$-form  $\omega \in \Lambda^1(\P, \LieH)$  is needed to define horizontality on $\P$. It satisfies,
\begin{itemize}
\item $\omega_p:V_p\P \rarrow \LieH$ is an isomorphism, $V_p\P$ being the vertical subspace of the tangent space $T_p\P$ at $p\in \P$.
\item $\R^*_h\omega=\Ad_{h\Em }\omega$, where  $\R_h$ is the right-action of the group. The connection is said \emph{pseudo-tensorial of type $(\Ad, \LieH)$}.
\end{itemize}
 The horizontal subspace complementary to $V_p\P$ in $T_p\P$ is defined by $H_p\P=\ker \omega_p$. The associated curvature $2$-form $\Omega\in\Lambda^2(\P, \LieH)$ is given by the Cartan structure equation: 
 \begin{align*}
 \Omega=d\omega +\tfrac{1}{2}[\omega, \omega]=d\omega+ \omega^2,
\end{align*}
 where $\omega^2=\omega \w \omega$ is the exterior product of forms, and the last equality holds when $\omega$ is matrix-valued. The curvature satisfies $\R^*_h\Omega=\Ad_{h\Em }\Omega$, and since it is horizontal (i.e it vanishes on $V\P$) it is said \emph{tensorial of type $(\Ad, \LieH)$}. Given a representation $(\rho, V)$ of $H$, an equivariant map $\Psi : \P \rarrow V$ satisfies $\R^*_h\Psi=\rho(h\Em )\Psi$. Since a map $\Psi \in C^\infty(\P)\simeq \Lambda^0(\P)$ is trivially horizontal, it is a tensorial $0$-form of type $(\rho, V)$. The covariant derivative of a tensorial $r$-form $\beta$ of type $(\rho, V)$ w.r.t $\omega$ is given by, $D^\omega\beta \defeq d\beta + \rho_*(\omega)\w\beta$. It is a $r+1$-tensorial form of type $(\rho, V)$. And $(D^\omega)^2\beta=\rho_*(\Omega)\w \beta$.
 
Given an open set $\U\subset \M$ and a local section $\sigma:\U \rarrow \P$, one can pull-back the above defined forms on $\M$. The local connection $1$-form $A=\sigma^*\omega \in \Lambda^1(\U, \LieH)$ describes the gauge potential. The local curvature $2$-form $F=\sigma^*\Omega \in \Lambda^2(\U, \LieH)$ describes the field strength of the potential. Locally the structure equation reads,
 \begin{align*}
 F=dA +\tfrac{1}{2}[A, A]=dA+A^2.
 \end{align*}
Matter fields are of the form $\psi=\sigma^*\Psi$, and their covariant derivatives are $D^A\psi=d\psi + \rho_*(A)\psi=\sigma^*D^\omega\Psi$. 
  
The gauge group of the principal bundle $\P$ is, 
\begin{align*}
\H=\big\{ \bgamma: \P \rarrow H \ |\ \R^*_h\bgamma(p) = h\Em \bgamma(p) h\big\}.
\end{align*}
It is isomorphic to the group  of vertical automorphisms of the principal bundle defined by,
\begin{multline*}
\Aut_V(\P)=\big\{ \Phi: \P \rarrow \P \ |\  \R^*_h\Phi(p)=\Phi(p)h,\\
                                                               \pi \circ \Phi(p)=\pi(p)=x\in \M \big\},
 \end{multline*}
  and whose composition law is $(\Phi_1 \Phi_2)(p) \defeq \Phi_2^*\Phi_1(p)=\Phi_1\circ\Phi_2(p)$. The isomorphism is given by $\Phi(p)=p\bgamma(p)$. 
The gauge group $\H$ acts geometrically on itself by
conjugacy action,
\begin{align}
\label{gauge group law}
\R_{\bgamma_2}^*\bgamma_1 \rdefeq \bgamma_1^{\bgamma_2^{}}=\bgamma_2\Em \bgamma_1^{}\bgamma_2^{}, \quad \bgamma_1^{}, \bgamma_2^{} \in \H,
\end{align}
and its action on global objects is
\begin{align}
\label{global activeGT}
\omega^{\bgamma}&=\bgamma\Em \omega \bgamma +\bgamma\Em d\bgamma,
\qquad \Omega^{\bgamma}=\bgamma\Em \Omega \bgamma, \notag\\[-2.5mm]
& \\[-2.5mm]
\Psi^{\bgamma}&= \rho(\bgamma\Em )\Psi, \qquad
(D^\omega\Psi)^{\bgamma}=\rho(\bgamma\Em )D^\omega\Psi. \notag
\end{align}
These are \emph{active gauge transformations}.\footnote{The \emph{passive gauge transformations} simply reflect another choice of local section. Given another open set $\U'\subset M$ with local  trivializing section $\sigma':\U'\rarrow \P$, such that on $\U'\cap \U$ on has $\sigma'=\sigma h$, with $h:\U'\cap\U \rarrow H$. One finds that the  pull-backs on each open set are related by $ A'=h^{-1}Ah +h\Em dh$, $F'=h^{-1}Fh$, $\psi'=\rho(h^{-1})\psi$ and $D^{A'}\psi'=\rho(h^{-1})D^A\psi$.
  This covers the case $\U'=\U$. Active and passive gauge transformations are formally similar. But notice that  passive gauge transformations relate two local descriptions of the \emph{same} global objects, while active gauge transformations relate \emph{different} global objects.} 
  The local version of the active gauge transformations reads, 
\begin{align}
\label{local activeGT}
A^\gamma&=\gamma\Em A\gamma +\gamma\Em d\gamma,\qquad  F^\gamma =\gamma\Em
F\gamma,& \notag\\[-2.5mm]
& \\[-2.5mm]
\psi^\gamma&=\rho(\gamma\Em )\psi,\qquad  (D^A\psi)^\gamma
=\rho(\gamma\Em)D^A\psi,& \notag
\end{align}
where $\gamma=\sigma^*\bgamma:\U\rarrow H$ belongs to the local gauge group $\H_\text{loc}(\U)$ over $\U$.

 We denote by $\A_\text{loc}(\U)$ the affine space of local connections on $\U$, on which the local gauge group $\H_\text{loc}(\U)$ acts. 

A gauge theory is specified by a Lagrangian form $L(A, \psi)$ which is required to be gauge-invariant: $L^\gamma=L$. Actually since the action is $S=\int_{\M} L$, the quasi-invariance of the Lagrangian form, $L^\gamma=L+da$, (i.e invariance up to a $d$-exact term), is enough to preserve the equations of motion given suitable boundary conditions, or if one works on a manifold $\M$ without boundary.

Let us stress the fact that in gauge field theories, any field is defined in a space on which acts the gauge group $\H$ of the theory (this action can be trivial). In other words, a field is more than a ``map between spaces'', since the action of $\H$ \emph{has to be} defined also. In this paper, all the actions of $\H$ will be induced by the \emph{geometrical} action defined by pull-back of objects.

\section{The dressing field method}  
\label{The dressing field method}  

The dressing field method \cite{GaugeInv} is a systematic way to reduce gauge symmetries. In its simplest version, it relies on the identification, within the theory, of a field defined as follows.

\begin{definition}\label{def-dressingfield}
A \emph{dressing field} is a map $\bu:\P \rarrow H$ with equivariance property $\R^*_h\bu=h\Em \bu$ and on which the action of the gauge group $\H$ is thus
\begin{align}
\label{global dressing law}
\bu^\bgamma=\bgamma \Em \bu.
\end{align}
\end{definition}

Given such a field $\bu$, one can define the following \emph{projectable} composite fields (i.e horizontal and gauge-invariant),
\begin{align}
\label{global comp fields}
\homega &\defeq \bu\Em \omega \bu + \bu\Em d\bu, \qquad \hOmega \defeq
\bu\Em \Omega \bu, \notag \\[-2.5mm]
& \\[-2.5mm]
\hPsi &\defeq \rho(\bu\Em )\Psi, \qquad \hD\hPsi \defeq\rho(\bu\Em )D^\omega\Psi.\notag
\end{align}
 The proof of the gauge invariance is straightforward. It is also easy to show that 
\begin{align*}
\hOmega = d\homega +\tfrac{1}{2}[\homega, \homega]=d\homega + \homega^2.
\end{align*}
This is the ``Main Lemma'' at the heart of the method presented in \cite{GaugeInv}.

 Given a trivializing section $\sigma:\U\rarrow \P$, we have the local dressing field, $u=\sigma^*\bu: \U \rarrow H$. The local gauge group $\H_\text{loc}(\U)$ acts as, 
\begin{align}
\label{local dressing law}
u^\gamma=\gamma\Em u.
\end{align}
The pull-backs of \eqref{global comp fields} give
 the \emph{gauge-invariant composite fields},
\begin{align}
\label{local comp fields}
\hA  &\defeq u\Em Au+u\Em du, \qquad \hF \defeq u\Em Fu, \notag \\[-2.5mm]
& \\[-2.5mm]
\hpsi &\defeq \rho(u\Em )\psi, \qquad \hD\hpsi \defeq \rho(u\Em )D^A\psi.\notag
\end{align}
These are known in the literature as \emph{generalized Dirac variables} \cite{Leader-Lorce, McMullan-Lavelle, Pervushin, Lantsman}. 
 Indeed Dirac \cite{Dirac55, Dirac58} pioneered the idea of working with invariant variables in QED proposing an explicit (non-local) realization of \eqref{local comp fields} in the abelian case $H=U(1)$. Thus the dressing field method, in its simplest version, is the geometrical foundation of the notion of Dirac variables. 

Notice that \eqref{local comp fields} are invariant under \emph{both} active \emph{and} passive gauge transformations. This means e.g that if one has $\hA'$ on $\U'$ and another has $\hA$ on $\U$, the two local descriptions agree on $\U'\cap\U$, $\hA'=\hA$. In other words $\hA$ is  \emph{globally defined} on $\M$. This is what it means for $\homega$ on $\P$ to be projectable: its projection $\hA$ on $\M$ is globally defined: $\pi^*\hA=\homega$. The same reasoning holds for $\hF$, $\hpsi$ and $\hD\hpsi$ or any dressed field.

\medskip
Let us stress that despite formal similarity \eqref{global comp fields} \emph{are not} active gauge transformations \eqref{global activeGT}. This is clear from the fact that owing to its transformation law \eqref{global dressing law}, different from \eqref{gauge group law}, the dressing field does not belong to the gauge group, $\bu\notin \H$. It is even more clear from the fact that e.g $\homega$ is not in the gauge orbit of $\omega$. Indeed $\homega$ is projectable, i.e is horizontal and gauge-invariant, so that it is not even a connection on $\P$. 

 In the same way \eqref{local comp fields} are not local description of the active gauge transformations \eqref{local activeGT} because $u\notin \H_\text{loc}(\U)$. In particular $\hA \notin \A_\text{loc}(\U)$ (see for instance Prop.~\ref{prop-homega-omega})  and the other gauge-invariant composite fields in \eqref{local comp fields} are not in the gauge orbits of $F$, $\psi$ and $D\psi$ respectively. This highlights the fact that the dressing field method is distinct from a gauge-fixing procedure. Indeed the latter consists in selecting a single representative in the gauge orbit of the gauge fields. Nonetheless the dressing field method could be a perfect substitute to the gauge-fixing. 

\medskip
If one can perform a change of field variables as \eqref{local comp fields}  in the theory, since the dressing if $H$-valued and due to the $\H_\text{loc}(\U)$-invariance of the Lagrangian, one has $L(A, \psi) = L(\hA, \hpsi)$.
Thus, one ends up with a Lagrangian form written in terms of gauge-invariant variables so that any equation or quantity derived from it will be automatically gauge-invariant as well. One can anticipate the benefit of this fact for the question of the nucleon spin decomposition.

Notice that even if one could find a gauge transformation $\gamma \in \H_\text{loc}(\U)$ such that $\gamma=u$, so that $A^\gamma=\hA$ and $\psi^\gamma=\hpsi$, these \emph{are obviously not} gauge-invariant equalities. Nevertheless it is true that $L(A^\gamma, \psi^\gamma)=L(A, \psi)= L(\hA, \hpsi)$.
This means that it may happen that a Lagrangian obtained through the dressing field method is mistaken for a gauge-fixed Lagrangian. We argued in \cite{Masson-Wallet, GaugeInv} that it is for example the case of the Lagrangian of the electroweak sector of the Standard Model in the so-called unitary gauge. The confusion is in this case especially prejudicial since the dressing field method provides  sensible interpretive shifts w.r.t the standard viewpoint.

\section{Geometric grounds of the Chen \& al. ansatz}  
\label{Geometric ground of the Chen et al. approach}  

\subsubsection*{The geometric ansatz} 

The Chen \& al. approach assumes the ansatz that the gauge potential splits as a ``pure-gauge'' and ``physical'' part: $A=A_\text{pure}+A_\text{phys}$. A splitting implicitly defined by the requirement that $F_\text{pure}=0$ and by the gauge transformations  $A_\text{pure}^\gamma =   \gamma\Em A_\text{pure} \gamma +\gamma\Em d\gamma$ and $A_\text{phys}^\gamma = \gamma\Em A_\text{phys} \gamma$. For the construction to make sense globally, we have also to assume that this pure gauge connection is \emph{globally} defined on the base manifold $\M$. Writing $A_\text{pure}=U_\text{pure} d U\Em_\text{pure}$, then $U_\text{pure}$ has been interpreted in  \cite{LorceGeomApproach} as a ``privileged basis in the internal space'', i.e a privileged point in the gauge orbit of any field. We will discuss this interpretation in the next section.

At the level of the principal fiber bundle $\P$, the Chen \& al. ansatz states 
 the existence of a (global) connection $\domega$ with vanishing curvature $\dOmega = 0$. 
Let us recall the following standard result of fiber bundle geometry: 
 \begin{thm}{ \cite[Corollary~9.2]{Kob-NomI}}\\ 
 \label{thm-trivialbundle} 
 Let $\P(\M, H)$ be a principal fiber bundle over a paracompact and simply connected manifold $\M$. Let  $\domega$ be a connection on $\P$ with vanishing curvature $\dOmega=0$. Then, there is a isomorphism of $H$-principal fiber bundles $\phi:\P \rarrow \M\times H$ such that $\domega = \phi^* \pi_H^*\omega_H$, where $\pi_H: \M\times H \rarrow H$ is the projection on the second factor, and $\omega_H$ is the Maurer-Cartan form on $H$.
\end{thm}

From now on, we will assume that $\M$ is paracompact and simply connected. The isomorphism  $\phi:\P \rarrow \M\times H$ of the theorem can be written as
\begin{align}
\label{isophi}
\phi(p)=\left( \pi(p),\ \bu\Em (p) \right),
\end{align}
where the map $\bu: \P \rarrow H$  displays the equivariance property $\R_h^*\bu=h\Em \bu$ in order to get $\phi(p h)=\phi(p)h$. Then a direct computation shows that $\domega = \bu d \bu\Em $. The existence of $\bu$ is related to the existence of $\domega$ (equivalently, of $A_\text{pure}$).  Let us now look at the action of the gauge group.

A gauge transformation $\domega \mapsto \domega^{\bgamma}$ is induced by an automorphism $\Phi \in \Aut_V(\P)$ and, with the notation 
\begin{equation}
\label{eq-gaugeactiononphi}
\phi^{\bgamma} = \Phi^* \phi = \phi \circ \Phi,
\end{equation}
one can check that $\domega^{\bgamma} = {\phi^{\bgamma}}^* \pi_H^*\omega_H$. In other words, the relation $\domega = \phi^* \pi_H^*\omega_H$ is $\H$-covariant for the respective natural geometrical actions of $\H$ on $\domega$ and on $\phi$. This leads naturally to define the (geometrical) right action of $\H$ on $\bu$ by $\phi^{\bgamma}(p) =  \left( \pi(p),\ ({\bu}^{\bgamma})\Em  (p) \right)$, and one gets ${\bu}^{\bgamma} = {\bgamma}\Em  {\bu}$. 

\medskip
To sum-up, the Chen \& al. ansatz is  equivalent to the following geometrical assumption, see Definition~\ref{def-dressingfield}: 

\emph{There is a dressing field $\bu : \P \rarrow H$.}

\smallskip

We have already shown that the Chen \& al. ansatz implies the existence of the dressing field $\bu$. Let us deduce the Chen \& al. ansatz from the dressing field. Following \cite[Proposition 2]{GaugeInv}, the existence of $\bu$ implies that the principal fiber bundle $\P$ is trivial, and \eqref{isophi} defines an explicit isomorphism. Let $\domega =  \bu d \bu\Em $. Then for any global section $\sigma : \M \rarrow \P$ (take $\sigma(x) = \phi\Em (x,e)$ for instance), $A_\text{pure} = \sigma^* \domega$ satisfies the requirements of the Chen \& al. ansatz.

Any connection $\omega$ on $\P$ can then be decomposed as $\omega = \domega + \alpha$ where $\alpha$ is a $(\Ad, \LieH)$-tensorial $1$-form. For any $\bgamma \in \H$, one then has $\omega^\bgamma = \domega^\bgamma + \alpha^\bgamma$ with $\domega^\bgamma = \bu^\bgamma d (\bu^\bgamma)\Em  = \bgamma\Em  \domega \bgamma + \bgamma\Em  d \bgamma$ and $\alpha^\bgamma = \bgamma\Em  \alpha \bgamma$. Then the curvature of $\omega$ is
\begin{align*}
\Omega& = d\omega + \omega^2 
= d\domega + d\alpha +\domega^2  +\domega\alpha +\alpha\domega+ \alpha^2,  \notag\\
& = d\alpha +[\domega, \alpha] + \alpha^2, \notag \\
& = D^\domega \alpha + \alpha^2,
\end{align*}
and the covariant derivative $D^\omega$ is
\begin{align*}
D^{\omega}\Psi & = d\Psi +\rho_*(\domega +\alpha) \Psi, \notag \\
&= d\Psi +\rho_*(\domega )\Psi +\rho_*(\alpha) \Psi, \notag\\
&= D^{\domega }\Psi + \rho_*(\alpha) \Psi.
\end{align*}
 Defining on $\M$ the fields $A \defeq \sigma^* \omega$, $A_\text{pure} \defeq \sigma^* \domega$, $A_\text{phys} \defeq \sigma^* \alpha$, $F \defeq\sigma^* \Omega$ and $D^A\psi \defeq \sigma^* (D^\omega \Psi)$, the above expressions pull-back as,
\begin{align}
 A &= A_\text{pure} + A_\text{phys} = udu\Em  + A_\text{phys},   \label{eq-local-split} \\
F&=D_\text{pure} A_\text{phys}+ A_\text{phys}^2 \notag\\
&= dA_\text{phys} + [udu\Em , A_\text{phys}] \ +\  A_\text{phys}^2, \notag  \\
D^A\psi&=D_\text{pure} \psi + \rho_*(A_\text{phys})\psi, \notag \\
&=d\psi +\rho_*(udu\Em )\psi \ + \ \rho_*(A_\text{phys})\psi, \notag
\end{align}
where $D_\text{pure}$ is the (local) covariant derivative associated to $A_\text{pure} = udu\Em$. We here use notations standard in the literature. With $\gamma \defeq \sigma^* \bgamma$, the gauge transformation of the gauge potential is
\begin{align}
A^\gamma&=A_\text{pure}^\gamma + A_\text{phys}^\gamma,\notag \\
&= (\gamma\Em A_\text{pure} \gamma +\gamma\Em d\gamma)\ + \ \gamma\Em A_\text{phys} \gamma      \label{eq-GT-split}
\end{align}
Equations \eqref{eq-local-split} and \eqref{eq-GT-split} reproduce the Chen \& al. ansatz. The Chen \& al. decomposition, and other decompositions it has inspired, for instance the Wakamatsu decomposition, make wide use of the objects $A_\text{pure}$, $A_\text{phys}$ and the operator $D_\text{pure}$ (called ``pure-gauge covariant derivative'').

 It is often stressed that in a gauge such that $A_\text{pure}=0$, the Chen \& al. decomposition reduces to the Jaffe-Monohar decomposition, so that the former is seen as a  ``gauge-invariant extension'' of the latter.  Since the Chen \& al. ansatz is equivalent to the existence of a dressing field $\bu$,  instead of merely fixing a gauge, one can suitably apply the dressing field method and define the ($\H$-gauge invariant) composite fields:
\begin{align*}
\homega &\defeq \bu\Em  \omega \bu + \bu\Em  d \bu \\
&= \bu\Em  (\bu d \bu\Em  + \alpha) \bu + \bu\Em  d \bu \\
&= \bu\Em  \alpha \bu \rdefeq \halpha, \\
\hOmega & \defeq \bar u\Em  \Omega \bar u, \\
\hPsi &\defeq \rho(\bar u\Em )\Psi, \\
\hD\hPsi &\defeq \rho(\bar u\Em )D^{\omega}\Psi= d\hPsi + \rho_*(\halpha )\hPsi .
\end{align*}
As fields on $\M$, let us define
\begin{align*}
\hA &\defeq \sigma^* \homega = \sigma^* \halpha = u\Em A_\text{phys} u\rdefeq \hA_\text{phys}, \\
\hF &\defeq \sigma^* \hOmega = d\hA_\text{phys} + \hA_\text{phys}^2 = u\Em F u, \\
\hpsi &\defeq \sigma^* \hPsi = \rho(u\Em )\psi,  \\
\hD\hpsi&=d\hpsi +\rho_*(\hA_\text{phys})\hpsi.
\end{align*}
   Notice that the composite field associated to $\domega$ vanishes, so that $\hA_\text{pure}=0$, as well as $\hA_\text{pure}^\gamma=0$ for any $\gamma = \sigma^* \bgamma$.

Now we can perform the change of variables in the Lagrangian form:
\begin{align*}
L(A, \psi) &= \tfrac{1}{2} \Tr(F\w *F) + \langle \psi, D\psi  \rangle, \\
&= \tfrac{1}{2} \Tr(\hF\w *\hF) + \langle\hpsi, \hD\hpsi  \rangle \\
& \rdefeq L(\hA_\text{phys}, \hpsi).
\end{align*}
Here $*$ is the Hodge star operator and $\langle\ ,\ \rangle$ is the inner product in the representation space $V$. 

  In the case of QCD, $V=\doubleC^3$ and $\langle \psi, \psi \rangle = \psi^\dagger\psi$. 
Starting with the Lagrangian $L_{QCD}(\hA_\text{phys}, \hpsi)$ it is now possible to write canonically a decomposition of the nucleon spin in every respect similar to the Jaffe-Monohar one. This decomposition thus displays a clear partonic interpretation and is gauge-invariant.

\subsubsection*{Arbitrariness in the choice\\ of the dressing field} 

Let $\Iso_V(\P, \M \times H)$ be the space of vertical principal fiber bundles isomorphisms $\P \rarrow \M \times H$. The group $\Aut_V(\P)$ 
 acts naturally on this space by \eqref{eq-gaugeactiononphi}. This action is transitive: for any $\phi_1, \phi_2 \in   \Iso_V(\P, \M \times H)$, $\Phi = \phi_2\Em  \circ \phi_1 \in \Aut_V(\P)$ 
  relates $\phi_1$ and $\phi_2$ by \eqref{eq-gaugeactiononphi}.
  
    Denote 
   $\Aut_V(\M \times H)$ the  group of vertical automorphisms of the trivial fiber bundle $\M \times H$. It acts on the right on $\Iso_V(\P, \M \times H)$ by $\phi \mapsto \Xi\Em  \circ \phi$ for any $\Xi \in \Aut_V(\M \times H)$ 
   considered as a map $\Xi : \M \times H \rarrow \M \times H$. This action is also transitive.

The gauge group of $\M\times H$ is thus, 
\begin{align*}
\H_0=\big\{ \bchi\!:\!\M\!\times\!H\!\rarrow\!H | \R^*_h\bchi(x, h')\!=\!h\Em\bchi(x, h')h  \big\}
\end{align*}
and its isomorphism with $\Aut_V(\M\times H)$ is given as usual by,  $\Xi(x, h) = (x, h \bchi(x,h))$. 
 
   Another description of $\Xi$ is given by a map $\bc : \M \times H \rarrow H$ defined by $\Xi(x, h) = (x, \bc(x,h) h)$. This map satisfies $\bc(x, h h') = \bc(x, h)$ for any $x \in \M$ and $h, h' \in H$ so that it depends only on $x$. We define $c : \M \rarrow H$ by $c(x) \defeq \bc(x,e)$. The induced action of $\Xi$ on $\bu$, written in terms of $c$, is given by $\bu_c(p) = \bu(p)\, c\circ \pi(p)$ for any $p \in \P$.

The choice of $A_\text{pure}$ in the original Chen \& al. ansatz is related to the choice of $\bu$ in our geometrical assumption by $A_\text{pure} = \sigma^* \domega = \sigma^*\bu d \bu\Em  = u d u\Em $ for $u = \sigma^*\bu$. The dressing field $\bu$ is completely characterized by its associated element $\phi \in \Iso_V(\P, \M \times H)$ given by \eqref{isophi}, so that the choice of $A_\text{pure}$ is related to the choice of a (global) trivialization $\phi$ of $\P$. This choice is not unique, and, given a fixed reference element $\phi_0 \in \Iso_V(\P, \M \times H)$, any other element can be obtained by the action of $\H$ or the action of $\H_0$ on $\phi_0$.

By construction, the composite field $\homega$ is $\H$-invariant, but it is not $\H_0$-invariant. For any $c \in \H_0$, one has
\begin{align}
\homega_c &=  \bu_c\Em  \omega \bu_c + \bu_c\Em  d \bu_c  \label{eq-actionconhomega}\\ 
&= c\Em  \bu\Em  \omega \bu c + c\Em  (\bu\Em  d \bu) c + c\Em  d c \notag \\
&= c\Em  \homega c + c\Em  d c,\notag \\
  \hOmega_c&=c\Em \hOmega c,     \notag \\
 \hPsi_c&=\rho(c\Em )\hPsi,   \qquad   \text{and} \qquad   \hD_c\hPsi_c=\rho(c\Em )\hD\hPsi.  \notag
\end{align}
On the other hand, one has
\begin{align}
\label{domega-c}
\domega_c=\bu_c d \bu_c\Em =\domega+ \bu(cdc\Em)\bu\Em,
\end{align}
and since $\omega$ is a connection on $\P$, it is invariant under $\H_0$,  $\omega_c=\omega$, so that
\begin{align}
\label{alpha-c}
\alpha_c=\alpha-\bu(cdc\Em)\bu\Em.
\end{align}

Let $\phi \in \Iso_V(\P, \M \times H)$ and $\bu : \P \rarrow H$ its associated map, which is a dressing field. Let $\sigma_e : \M \rarrow \M \times H$ be the section $\sigma_e(x) \defeq (x, e)$. 

\begin{prop}
\label{prop-homega-omega}
As a $1$-form on $\P$, one has
\begin{equation*}
\homega = (\sigma_e \circ \pi)^* {\phi\Em }^* \omega .
\end{equation*}
In other words, $\homega$ is related by $(\sigma_e \circ \pi)^*$ to the connection $1$-form $\omega_0 \defeq {\phi\Em }^* \omega$ on the trivial bundle $\M \times H$.

Let $A_0 \defeq \sigma_e^* \omega_0$. Then $\hA = A_0$ on $\M$.
\end{prop}

\begin{proof}
Let $f_\bu : \P \rarrow \P$ be defined by $f_\bu(p) \defeq p \bu(p)$. Then one has
\begin{align*}
\phi(f_\bu(p)) &= \phi( p \bu(p)) 
= (\pi(p), \bu(p \bu(p))\Em ) \\
&= (\pi(p), \bu(p)\Em \bu(p)  ) 
= (\pi(p), e)
\\
&= \sigma_e \circ \pi(p),
\end{align*}
which implies $f_\bu = \phi\Em  \circ \sigma_e \circ \pi$. It is shown in \cite[p.~8]{GaugeInv} that $\homega = f_\bu^* \omega$, so that $\homega = (\phi\Em  \circ \sigma_e \circ \pi)^* \omega = (\sigma_e \circ \pi)^* {\phi\Em }^* \omega$.

We have $\hA = \sigma^* \homega = \sigma^* \pi^* \sigma_e^* \omega_0 = (\pi \circ \sigma)^* A_0 = A_0$ since $\pi \circ \sigma = \id_\M$.
\end{proof}

Since $\homega=\halpha$, $\hA$ depends only on the tensorial part of $\omega$. This is supported by the fact that it can be shown that $ (\phi\Em  \circ \sigma_e)^*\domega=0$ so that the pure part of the connection disappears completely because $d (\bu\Em \circ \phi\Em  \circ \sigma_e)= d (x\mapsto e) = 0$.

In this proof, we have shown that the following diagram is commutative:
\begin{equation*}
\xymatrix{ 
{(\P, \homega)} \ar[rr]^-{f_\bu} \ar[dr]_-{\sigma_e \circ \pi} & {} & {(\P, \omega)} \ar[dl]^-{\phi} \\
{}  & {(\M \times H, \omega_0)} & {}
}
\end{equation*}
The action \eqref{eq-actionconhomega} of $\H_0$ on $\homega$ is nothing else that the (natural) action of $\H_0$ on $\omega_0$ (as a connection $1$-form on $\M \times H$). Notice also that the equality $\hA = A_0$ shows that $\hA \notin\A_\text{loc}(\M)$, and that $\hA$ belongs to the affine space of local connections on $\M \times H$. This shows a twofold feature of the field $\hA$: It is $\H_\text{loc}(\U)$-invariant but a  $\H_0$-connection. The triviality of $\P$ gives in some extent a ``swing effect'' and morally $\hA$ remains a gauge field.

From this result, we conclude that the arbitrariness in the choice of the dressing field $\bu$ (or in the choice of $A_\text{pure}$ in the original Chen \& al. ansatz) is related to the transitive action of the gauge group $\H_0$ of $\M \times H$, which is isomorphic to $\H$.\footnote{Given a fixed element $\phi_0 \in \Iso_V(\P, \M \times H)$, the map $\tphi_0 : \H \rarrow \H_0$ defined by $\tphi_0(\Phi) \defeq \phi_0\Em \circ \Phi \circ \phi_0$ is an isomorphism.} This arbitrariness has been interpreted as a ``Stueckelberg symmetry''  in \cite{Stoilov, LorceGeomApproach}, see discussion below.

Another point of view is to look at the action of $c$ on $\hA = A_0$ as the transformation induced by a change of local trivialization of $\M \times H$, from the (global) section $\sigma_e$ to the new section $x \mapsto (x, c(x))$.

\section{Discussion} 
\label{Discussion} 

In this section we stress how the present work contributes to the clarification of some questions within the literature seeded by  the Chen \& al. approach to the nucleon spin decomposition.

The Lorentz covariance of the initial Chen \& al. ansatz has been early questioned. Subsequently much energy has been deployed in order to settle the question, and now the consensus seems to agree on the Lorentz covariance of the approach (see \cite{LorceGeomApproach, Leader-Lorce, Wakamatsu11}).  Here we have provided a sound differential geometric basis to Chen \& al. ansatz, the object involved are thus differential forms which are intrinsically defined. This secures \emph{general relativistic} covariance. Indeed the base manifold $\M$ on which we localized our construction is an arbitrary manifold, not necessarily reduced to the Minkowski space. In our view this is an improvement that easily circumvent much of the concerns and discussions around the question of the Lorentz covariance. 

We have shown that at a global level, the Chen \& al. anzats amounts to assume tacitly the triviality of the underlying principle bundle $\P$. This in turn is equivalent to the existence of a dressing field as defined in \cite{GaugeInv}. The dressing field method reproduces transparently the construction first proposed in the geometric section of \cite{LorceGeomApproach}.  However in the literature the status of the dressing field $u$ (often denoted $U_\text{pure}$) is not always clear. On the one hand we find in \cite{Leader-Lorce, LorceCanForm} caveats as to not mistake  expressions like \eqref{local comp fields} for gauge transformations. On the other hand it is often said \cite{LorceGeomApproach, Leader-Lorce, LorceWilson, LorceCanForm} that $u$ specifies a ``privileged basis'' in the internal space (admittedly, in our framework $u$ defines at most a ``privileged trivialization'' of $\P$).
 As a matter of fact in \cite{LorceGeomApproach, LorceWilson}  expressions like \eqref{local comp fields} are referred to as fields in the ``natural gauge/basis''. This could be understood as suggesting that $u$ is a gauge transformation sending the gauge fields into  privileged points of their respective gauge orbits. 

Here we clarify this point: $u$ being a dressing field it does not belong to the gauge group $\H_\text{loc}(\U)$, so \eqref{local comp fields} are not gauge transformations and none of these composite fields belong to the gauge orbits of the initial gauge fields. Recall e.g that $\hA$ is not even a connection, $\hA \notin \A_\text{loc}(\U)$.

\medskip
The most delicate point to discuss is the arbitrariness in the choice of the dressing field 
which yields transformations of the initial gauge variables \eqref{domega-c}-\eqref{alpha-c} and of the dressed variables~\eqref{eq-actionconhomega}. The local version of these equations are usually referred to as a ``Stueckelberg symmetry'' (see Appendix).
Actually we've shown that this arbitrariness is controlled by the gauge group $\H_0$  of the trivial bundle $\M\times H$, which is isomorphic to the initial gauge group $\H$. So the situation after the dressing operation mirrors the situation before it. Indeed, before dressing, $\H_\text{loc}(\U)$ acts on the set of original variables 
  by \eqref{local activeGT}, but $\H_{0, \text{loc}}(\U)$ acts trivially. And after dressing, $\H_\text{loc}(\U)$ acts trivially on the set of dressed variables, but $\H_{0, \text{loc}}(\U)$ acts by (the local version of) \eqref{eq-actionconhomega} which mimics \eqref{local activeGT}. See table~\ref{table1}. 

\begin{center}
\begin{tabular}{cc}
\toprule
	 $\H_\text{loc}(\U)$	  & 	$\H_{0, \text{loc}}(\U)$   \\
\midrule
	$A^\gamma=\gamma^{-1}A\gamma+\gamma\Em d\gamma$	  & 	$A_c=A$      \\
    $F^\gamma=\gamma^{-1}F\gamma $					  &          $F_c=F$       \\ 
  $\psi^\gamma=\rho(\gamma^{-1})\psi$   				  &    $\psi_c=\psi$    \\ 
 $(D\psi)^\gamma=\rho(\gamma^{-1})D\psi$ 				  &    $D\psi_c=D\psi$     \\ 
\midrule
	$\hA^\gamma=\hA$ 				  & 	$\hA_c=c^{-1}\hA c +c\Em dc$      \\
    $\hF^\gamma=\hF $					  &          $\hF_c=c^{-1}Fc$       \\ 
  $\hpsi^\gamma=\hpsi$   				  &    $\hpsi_c=\rho(c^{-1})\hpsi$    \\ 
 $(\hD\hpsi)^\gamma=\hD\hpsi$ 				  &    $\hD\hpsi_c=\rho(c^{-1})\hD\hpsi$     \\ 
\bottomrule
\end{tabular}
\captionof{table}{Action of the gauge groups $\H_\text{loc}(\U)$ and $\H_{0, \text{loc}}(\U)$ on the local fields before and after dressing.}
\label{table1}
\end{center}

Of course the Lagrangian after dressing, $L(\hA, \hpsi)$, is invariant under $\H_{0, \text{loc}}(\U)$, but the gauge-invariant Jaffe-Monohar decomposition of the nucleon spin obtained from it, is not. So one faces the same problem one had with the original gauge symmetry $\H_\text{loc}(\U)$. 

To be clear, what is usually referred to as a ``Stueckelberg symmetry'' is actually the original gauge symmetry in disguise. Then, nothing is really gained by performing the Chen \& al. split, and/or by applying the dressing field method here. The seemingly erased gauge symmetry, 
 eventually reappears. 

The reason is actually simple: one cannot hope to produce something from nothing. With just the connection $\omega$, i.e the gauge potential $A$, one lacks the other degrees of freedom necessary to genuinely neutralize the gauge symmetry. 
In order to make the dressing field method really effective, we have observed that the dressing field $u$ should be ``extracted'' from an auxiliary field \emph{already given}  in the whole theory at hand. This was indeed the case for the examples treated in \cite{GaugeInv}: in the case of the electroweak sector of the standard model $u$ is extracted from the scalar field $\varphi$; in the case of General Relativity, $u$ is the vielbein extracted from the soldering form $\theta$ which is a component of a Cartan connection. 
This is also  what does the field $B$, introduced beside the gauge potential $A$, in the Stueckelberg trick (see Appendix).
 In these examples the arbitrariness in the choice of the dressing field is not isomorphic to the gauge group: either it is an interesting symmetry (the coordinate changes in the case of gravitation) or it can be drastically reduced by an additional physical consideration (in the case of the electroweak sector of the standard model, see \cite{Masson-Wallet}). For the Stueckelberg trick it is even reduced to nothing. 
 
 Concerning the problem of the nucleon spin decomposition, all this discussion suggests to find out a natural candidate for an auxiliary field from which a relevant dressing field ought to be extracted. This issue is out of the scope of the present paper.

\section*{Conclusion}  
\label{Conclusion}  

 In this letter we give a reconstruction of the Chen \& al. ansatz using a differential geometric framework which is economic, transparent and secures general relativistic covariance.  In our view this settles the debate about the Lorentz covariance of the original Chen \& al. approach. 
 
   We propose a 
global geometrical 
framework for the Chen \& al. ansatz. We show that this ansatz implies the existence of a dressing field as defined in \cite{GaugeInv}. Then, through the dressing field method, it becomes possible to rewrite the Lagrangian of the theory with a priori gauge-invariant variables (see bottom part of Table \ref{table1}). Accordingly, a Jaffe-Monohar-type decomposition of the nucleon spin which is thus seemingly gauge-invariant, is readily obtained from the new Lagrangian.
 
  However, we clarify the fact that what is usually referred to as a ``Stueckelberg symmetry'' is actually an avatar of the initial gauge symmetry. Hence nothing is really gained in the original Chen \& al. approach, unless one is able to extract the dressing field from an auxiliary field providing the degrees of freedom necessary to genuinely neutralize the gauge symmetry.

\appendix 

\section*{Appendix} 
\label{Appendix}  

The terminology ``Stueckelberg symmetry''  has been used in \cite{Stoilov, LorceGeomApproach} to qualify the arbitrariness in the splitting of the gauge potential $A$ in the Chen \& al. approach. It stems from a formal analogy with the Stueckelberg trick which aims at implementing a gauge symmetry where there is none by introducing a field $B$ with adequate transformation law. Take the prototype $U(1)$ abelian Stueckelberg Lagrangian,
 \begin{align*}
 L(A, B)=\tfrac{1}{2}F\!\w\! *F +   m^2 (A\!-\!\tfrac{1}{m}dB)\!\w\!* (A\!-\!\tfrac{1}{m}dB).
 \end{align*}
 The gauge transformations of the potential $A$ and of the Stueckelberg field $B$ are defined by,
 \begin{align*}
 A^\gamma=A-d\theta, \qquad \text{and} \qquad  B^\gamma=B-m\theta,
 \end{align*}
 so that the variable $A'=A-\tfrac{1}{m}dB$ is gauge invariant. By the way, the abelian case of the above construction gives $A_\text{pure}=udu\Em=-d\chi_\text{pure}$,  $cdc\Em=-d\xi$ and
 \begin{align*}
 A_{\text{phys}, c}= A_\text{phys}+d\xi, \quad  \text{and} \quad \chi_{\text{pure}, c}=\chi_\text{pure} +\xi,
 \end{align*}
 so that $A=A_\text{phys}-d\chi_\text{pure}$ is preserved. 
 One is tempted to formally identify $A_\text{phys}$ with $A$. But according to the main text, this is misleading since the former is tensorial and the latter is a gauge potential. However $\chi_\text{pure}$ plays the role  of the Stueckelberg field $B$.
 
 Actually the transformation for the Stueckelberg field $B$ is precisely an abelian version of \eqref{local dressing law}, so that $B$ is indeed a local dressing field. This means that $A'$ and $F'=F$ are abelian instances of \eqref{local comp fields}. Applying the dressing field method on the Stueckelberg Lagrangian one has
 \begin{align*} 
 L(A, B)=L(A')=\tfrac{1}{2}F'\w*F' + m^2A'\w A'.
 \end{align*}
 This Lagrangian written in terms of gauge invariant variables describes a theory where the $U(1)$ gauge symmetry has been completely factorized out. To the extent that a parallel between the Stueckelberg trick and the Chen \& al. approach can be drawn, it stems from the observation that both are tightly related to the dressing field method. 
 
  A more detailed discussion between the Stueckelberg trick and the dressing field method has been given in~\cite{GaugeInv}.

\bibliography{Biblio}

\end{multicols}

\end{document}